\documentclass{article}
\usepackage{amsmath, amsthm, amssymb,graphicx,subfig,multirow}
\newcommand{\cond}[1]{\mathop{\rm{cond}}(#1)}

\newcommand{\icond}[1]{\mathop{\rm{icond}}(#1)}

\newcommand{\norm}[1]{\left\|#1\right\|} 
\newcommand{\inv}[1]{#1^{-1}}
\newcommand{\mpinv}[1]{#1^{\dagger}}  
\newcommand\eps{\epsilon}
\newcommand\eref[1]{$(\ref{#1})$}

\newcommand\R{\mathbb{R}}

\newtheorem{theorem}{Theorem}[section]
\newtheorem{lemma}[theorem]{Lemma}
\newcommand{\rank}[1]{\mathop{\rm{rank}}(#1)}

\title{A condition number analysis of an algorithm for solving a system of polynomial equations with one degree of freedom\thanks{Supported in part by
NSF DMS 0434338, NSF CCF 0085969, and a grant from NSERC
(Canada).}}

\author{Gun Srijuntongsiri\thanks{Sirindhorn International Institute of Technology, Thammasat University, 131 Moo 5, Tiwanont Road, Bangkadi,
Muang, Pathum Thani, 12000, Thailand. Email: gun@siit.tu.ac.th.} \and Stephen A. Vavasis\thanks{Department of Combinatorics and Optimization,
University of Waterloo, 200 University Avenue W., Waterloo, ON N2L
3G1, Canada. Email: vavasis@math.uwaterloo.ca.}}

\begin{document}

\maketitle

\begin{abstract}
This article considers the problem of solving a 
system of $n$ real polynomial equations in $n+1$ variables.  We propose an algorithm
based on Newton's
method and subdivision for this problem.  Our algorithm is intended only for
nondegenerate cases, in which case the solution is a 1-dimensional curve.
Our first main contribution is a definition of a condition number 
measuring reciprocal distance to degeneracy
that can distinguish poor and well conditioned instances of this problem.  
(Degenerate problems would be infinitely
ill conditioned in our framework.)
Our second contribution, which is the main
novelty of our algorithm, is an analysis showing that its running
time is bounded in terms of the condition number of the
problem instance as well as $n$ and the polynomial degrees.  
\end{abstract}

\section{Introduction}
We consider the problem of finding all zeros of a polynomial function
$f:[0,1]^{n+1}\rightarrow \R^n$.  The zero-set of such a function will, in
the generic case, be a 1-dimensional algebraic set.  Our algorithm is
enumerative in nature, and therefore is feasible only in the case of small
values of $n$. 
We refer to this problem as the {\em single-degree-of-freedom polynomial 
system problem} (SDPS).

Perhaps the most common application of the SDPS problem is finding the intersection
of two rational or polynomial surfaces, the so-called surface/surface intersection
(SSI) problem.  In this case, one is given two polynomials $p_1,p_2$ both mapping
$\R^2$ to $\R^3$ that each parametrize a surfaces.  
The problem is to find their intersection,
i.e., all points $(s,t,u,v)$ in $[0,1]^4$ such that $p_1(s,t)-p_2(u,v)=0$.
Other applications arise in robotics and motion planning.  A final application
is global optimization in which one finds all the local minimizers by constructing
a network of paths that connect local minimizers (see, e.g., \cite{Leary}).

Our proposed algorithm is a hybrid between subdivision and iterative
methods.  This hybrid idea has been used to solve surface/surface
intersection by Koparkar \cite{koparkar} and to solve line/surface
intersection by Toth \cite{toth}.  The approach is to subdivide the
domain recursively, discard the subdomains found to contain no
solutions, and invoke an iterative method to locate a solution once it
is certain that the iterative method converges.  The convergence tests
used by Toth and Koparkar are both based on contraction mapping and
evaluating ranges of functions.

Our first main contribution, detailed in Section~\ref{section_cond}, is the
definition of a condition number for SDPS problems.
Intuitively, a problem instance is ill conditioned if it 
is close to a degenerate instance.
A degenerate instance is one in which the Jacobian fails to have full rank at 
a root.  Our condition number is the reciprocal of a quantity
related to nearness to
degeneracy.  It is natural to expect that algorithms would have poorer behavior
as the condition number grows larger.

Our algorithm, which is presented in Sections~\ref{sec:algo}--\ref{sec:impdetails},
is similar to Koparkar's in that it subdivides the
parametric domains of the problem until the subdomains pass
certain tests. It uses a bounding volume of a subdomain to exclude
any that cannot have a solution.  Our convergence test is
based on the Kantorovich theorem, which tells us if Newton's
method converges quadratically for the initial point in question
in addition to whether it converges at all.  
For this reason, we
can choose to hold off Newton's method until quadratic convergence
is assured. Kantorovich's theorem is presented in Section~\ref{sec:kanto}.

The main feature of our algorithm is that there is a lower bound on
the size of the smallest hypercube occurring during the course of the
algorithm that depends on the reciprocal condition number of the
problem instance and on the polynomial degrees.  Because our interest is
in the low-dimensional and moderate degree case, we regard the factors
depending on the degrees as `constants' and the dependence on the
condition number as the interesting feature.
This analysis is presented in
Section~\ref{section_analysis}.  A lower bound on the smallest hypercube
size consequently implies an upper bound on the overall running time.  

As mentioned above, the SSI problem is a special case of SDPS with $n=3$.  Since there are many algorithms for SSI proposed in literature, we discuss here the main advantage of our algorithm as compared to other SSI algorithms.  To the best of our knowledge, there is no previous algorithm for SSI in this class whose running time has been bounded in terms of the condition of
the underlying problem instance, and we are not sure whether such an
analysis is possible for previous algorithms.  Indeed, we do not know
of any SSI algorithm in the literature that has any {\em a priori}
bound on the running time.  We do know, however, that some algorithms
do not have this property---their running time can be arbitrarily
large even if the input instance is well conditioned.  This is because
these algorithms can sometimes create degenerate or nearly degenerate
subproblems even though the original input instance is well
conditioned.  Section \ref{section_other} shows in details how
marching methods based on collinear normal points and Koparkar's
algorithm in particular have the capacity to create bad subproblems
from a good instance.  

The notion of bounding the running time of an iterative method in
terms of the condition number of the instance is an old one, with
the most notable example being the condition-number bound of
conjugate gradient (see Chapter 10 of \cite{gvl}). This approach
has also been used in interior-point methods for linear
programming \cite{freund}, Krylov-space eigenvalue computation
\cite{toh}, and the line/surface intersection problem
\cite{srijuntongsiri_lsi}.  We note that a related problem of computing convex hull of points on a plane is shown to always be well-conditioned \cite{jiang}.

An additional motivation, not pursued further herein, for defining a
condition number and condition-aware algorithms like ours is that this
creates the possibility of preconditioning.  {\em Preconditioning},
which has been very successfully applied in numerical linear algebra
(see, e.g., \cite{TrefethenBau}), means improving the condition number
of an instance via some kind of transformation prior to solving it.

We now define the problem under consideration more precisely by
specifying a representation for the input polynomial system.
Let $Z_{i,m}(t)$ denote the \emph{Bernstein polynomials}
\[
Z_{i,m}(t) = \frac{m!}{i!(m-i)!}(1-t)^{m-i}t^i.
\]
We are interested in finding \emph{all} points $x=\left(x_1,x_2,\ldots,x_{n+1}\right)^T \in [0,1]^{n+1}$ satisfying
\begin{equation}
\label{maineq}
f(x) \equiv \sum_{i_1 = 0}^{m_1}\cdots \sum_{i_{n+1}=0}^{m_{n+1}} b_{i_1,  \cdots, i_{n+1}} Z_{i_1,m_1}(x_1) \cdots Z_{i_{n+1},m_{n+1}}(x_{n+1}) = 0,
\end{equation}
where $b_{i_1, \cdots, i_{n+1}} \in \mathbb{R}^n$ $(i_j = 0, 1, \ldots, m_j)$ denote the coefficients, also known as the \emph{control points}.  Therefore, the
problem is specified by the $(m_1+1)(m_2+1)\cdots (m_{n+1}+1)$ control points.  
(See a further remark on this matter in Section~\ref{section_conclusion}).
This form of a multivariate polynomial is sometimes called {\em tensor product} B\'ezier
representation: it presumes 
that the maximum degree of variable $x_i$ (separately) is $m_i$
for each $i=1,\ldots,n+1$.  
Note that $f$ is a function that maps $\mathbb{R}^{n+1}$ to $\mathbb{R}^n$.  
Note that this representation would be intractable
for a large value of $n$, but, as mentioned earlier, our algorithm is
intended for small values such as $n=3$.

The Bernstein basis is known to 
have better numerical stability for polynomials on the unit interval
than the power basis \cite{farouki_rajan, farouki_stable_bernstein},  
and computation using parametric representation is often
much more efficient than other types of surface representations.  
Furthermore,
our algorithm makes direct use of the Bernstein-B\'ezier representation.
In particular, our exclusion test is based
on B\'ezier control points.  It
should be noted that the algorithm proposed in this article can be
generalized to use with parametric surfaces represented by other
polynomial bases provided that an appropriate exclusion test is
available, and a few other properties hold for the basis.  Refer to
\cite{srijuntongsiri_lsi} for a related algorithm for line/surface
intersection that can operate on parametric surfaces represented by
other polynomial bases.

\section{The theorem of Kantorovich}
\label{sec:kanto}

Denote the closed ball centered at $x$ with radius $r>0$ by
\[
\bar{B}(x,r) = \{ y \in \mathbb{R}^n : \norm{y-x} \leq r \},
\]
and let $B(x,r)$ denote the interior of $\bar{B}(x,r)$.
Kantorovich's theorem in affinely invariant form, which is valid for
any norm, is as follows.

\begin{theorem}[Kantorovich, affinely invariant form \cite{deuflhard,kantorovich}]
\label{standardkantorovich}
Let $f : D \subseteq \mathbb{R}^n \rightarrow \mathbb{R}^n$ be differentiable in
the open convex set $D$. Assume that for some point $x^0 \in D$, the Jacobian $f'(x^0)$
is invertible with
\[
\norm{f'(x^0)^{-1}f(x^0)} \leq \eta.
\]
Let there be a Lipschitz constant $\omega > 0$ for $f'(x^0)^{-1} f'$ such that
\[
\norm{f'(x^0)^{-1}(f'(x)-f'(y))} \leq \omega \cdot \norm{x-y} \textrm{ for all } x,y \in D.
\]
If $h = \eta\omega \leq 1/2$ and $\bar{B}(x^0,\rho_-) \subseteq D$, where
\[
\rho_- = \frac{1-\sqrt{1-2h}}{\omega},
\]
then $f$ has a zero $x^*$ in $\bar{B}(x^0,\rho_-)$. Moreover, this zero is the unique zero
of $f$ in $(\bar{B}(x^0,\rho_-) \cup B(x^0,\rho_+)) \cap D$ where
\[
\rho_+ = \frac{1+\sqrt{1-2h}}{\omega}
\]
and the Newton iterates $x^k$ defined by
\[
x^{k+1} = x^k - f'(x^k)^{-1}f(x^k)
\]
are well-defined, remain in $\bar{B}(x^0,\rho_-)$, and converge to $x^*$. In addition,
\[
\norm{x^*-x^k} \leq \frac{\eta}{h}\left( \frac{(1-\sqrt{1-2h})^{2^k}}{2^k} \right), k = 0,1,2,\ldots
\]
\end{theorem}
We call $x^0$ a \emph{fast starting point} if the quantity $h$
defined above satisfies $h \leq 1/4$ and $\bar{B}(x^0,\rho_-)
\subseteq D$.  In this case, quadratic convergence of the iterates
starting from $x^0$ is implied.

\section{A condition number of a polynomial system with one degree of freedom}
\label{section_cond}

In this section we propose a definition for a condition number of the SDPS problem and prove that our condition number is related to the distance
from degeneracy.  Let $M$ denote the maximum $p$-norm among the control points of $f$ for some $p$.  
Later on, we will specialize to the infinity norm.
It is easy to show that this quantity satisfies
the axioms of a norm, so we will write $M$ also as $\Vert f\Vert$.
Define the {\em condition number} of $f$ to be
\begin{equation}
\label{cond_def} \cond{f} = M\cdot \max_{x \in [0,1]^{n+1}} \left( \min
\left\{{\frac{1}{\norm{f(x)}}}, \norm{\mpinv{f'(x)}}\right\}
\right).
\end{equation}
Here, the notation $A^\dagger$
means $A^\dagger = A^T(AA^T)^{-1}$.  In the case that
the rank of $A$ is $n$, this corresponds to the Moore-Penrose
pseudo-inverse of $A$, typically denoted as $A^+$.  In general,
the Moore-Penrose pseudo-inverse is defined for matrices of all ranks.
In this paper, however, we need $A^\dagger$ only in the case that
$\rank(A)=n$; we will 
take the second factor $\norm{\mpinv{f'(x)}}$ of \eref{cond_def}
to be $\infty$ if the rank of $f'(x)$ is
less than $n$.
Similarly, the first factor is taken to be $\infty$ if $f(x)=0$.
Note that small condition number means the problem is
well-conditioned.
Indeed, it follows
from \eref{eq:fMbound} and \eref{eq:Mpinvineq} below
that both terms in the min have a lower bound of  $\mbox{const}/M$.

The rationale for this definition is that, as mentioned earlier,
a degenerate instance has a point $x$ such that $f(x)=0$ and
$\rank(f'(x))<n$.  For such a point, both terms occurring in the min
of \eref{cond_def}
are infinity, i.e., the condition number is infinite.  Thus, the
problem is ill-conditioned
if $f(x)$ is close to zero and $f'(x)$ is close
to rank-deficiency at the same point $x$.

The inclusion of the factor
of $M$ makes the definition scale-invariant.  Note that $M$ depends
on the choice of basis, namely, tensor-product Bernstein-B\'ezier basis,
whereas the other factors in the condition number
are basis-independent.  The dependence on
basis, however, is only up to a scalar factor depending
on degree.  This is because polynomials are a finite-dimensional vector
space, hence all norms are equivalent.  We could redefine $M$ in a 
basis-independent manner as follows:
$$M_{\rm BI}=\max_{x\in[0,1]^{n+1}} \Vert f(x)\Vert.$$
It follows from \eref{eq:fMbound} and \eref{theta_eq}
below that $M_{\rm BI}$ and $M$ differ by a scalar
that is bounded in terms of the degrees.
The definition $M_{\rm BI}$, however, 
would be harder to compute
in practice.

This condition number is similar to the definition of
$\tilde \kappa(f)$  of Cucker et al.\ \cite{Cuckeretal} 
for the problem of computing isolated roots of polynomial
systems $\R^{n}\rightarrow\R^n$.
Our definition is somewhat simpler than theirs, however,
because
of the assumption we have imposed
that the domain of interest
is $[0,1]^{n+1}$ rather than all of $\mathbb{R}^{n+1}$.  This simplifying
assumption obviates the need for introducing projective
space and scaling of coefficients as in \cite{Cuckeretal}.

The classical Turing Theorem \cite{Chaitin} states that the condition
number of a matrix, which bounds the iteration count of the conjugate
gradient algorithm, is exactly the reciprocal of the relative distance
of the matrix to singularity.  Similarly, Shub and Smale show that
their condition number for a homogeneous polynomial system is equal to
the distance of the system to singularity \cite{shub}.  We now derive
a result showing that our proposed condition number is also related
to
the distance of an SDPS instance to degeneracy.

Before stating and proving the theorem, we
require two well known bounds concerning polynomials 
in Bernstein-B\'ezier form:
\begin{equation}
\Vert f(x)\Vert \le \Vert f\Vert
\label{eq:fMbound}
\end{equation}
for all $x\in[0,1]^{n+1}$.
This follows because every value of $f$ over the
parametric domain is a convex combination
of control points \cite{Farin}.
Next, 
\begin{equation}
\Vert f'(x)\Vert_\infty \le \norm{f'}_\infty \leq 2(n+1)\max(m_1, m_2, \ldots, m_{n+1})\Vert f\Vert_\infty
\label{eq:fpMbound}
\end{equation}
for all $x\in[0,1]^{n+1}$.
This follows because
the control points of a column of $f'$ (i.e., a partial derivative of $f$)
are finite differences of control points
of $f$ multiplied by the degree in the direction of differentiation
as in \eref{eq:bezderiv1}.  A
factor of $2$ comes from the taking of finite differences, and a further
factor of $n+1$ arises from the fact that the infinity norm is a sum
over rows (not columns) of the derivative.  Furthermore, applying the equivalence of norms to \eref{eq:fpMbound} yields
\begin{equation}
\label{eq:ebound}
\norm{f'} \leq 2\alpha_n(n+1)\max(m_1, m_2, \ldots, m_{n+1})\Vert f\Vert
\end{equation}
for any arbitrary $p$-norm, where $\alpha_n > 0$ is a scalar depending on $n$ and the choice of norm.

Now, finally, we come to the main theorem of this
section.
\begin{theorem}
Let $f:[0,1]^{n+1}\rightarrow\mathbb{R}^{n}$ be a polynomial function of degrees 
$m=(m_1,m_2,\ldots, m_{n+1})$ in its $n+1$ variables, and assume that it
is nondegenerate, i.e., there is no $x\in[0,1]^{n+1}$ such that
$f(x)=0$ and $\rank(f'(x))<n$.
Then any $\tilde f$ satisfying
\begin{equation}
\frac{\Vert \tilde f-f\Vert}{\Vert f\Vert} \le \frac{c_m}{\cond{f}}
\label{eq:hatfminusf1}
\end{equation}
is nondegenerate, where $c_m$ is a scalar depending on the degrees and the choice of norm.

Conversely,
there exists a degenerate 
polynomial $\hat f$ such that
\begin{equation}
\frac{\Vert \hat f - f\Vert}{\Vert f\Vert}
\le \frac{ c_m'}{\cond{f}},
\label{eq:hatfminusf2}
\end{equation}
where
$c_m'$ is another scalar depending on the degrees.
\label{thm:disttodegen}
\end{theorem}

\noindent
{\bf Remark 1.} 
As mentioned above, the vector and matrix norms appearing in 
\eref{eq:hatfminusf1}, \eref{eq:hatfminusf2}
may be any of the standard $p$-norms, although later we will specialize
to the infinity norm.
Inequalities \eref{eq:hatfminusf1} and \eref{eq:hatfminusf2} involve
the norm of the polynomial function.  
As mentioned above, we take this to mean the maximum norm control point
when written in Bernstein-B\'ezier form.  

\noindent
{\bf Remark 2.}  Let $\mathcal{D}$ denote the set of 
degenerate polynomials, that is, those polynomials $f$
such that $\cond{f}=\infty$.
This theorem shows that 
our condition number is, up to constant factors,
the reciprocal distance of
$f$ to $\mathcal{D}$ scaled by the norm of $f$.
Consider a larger set $\mathcal{D}'$
of degenerate polynomials defined as follows.
Polynomial $f\in\mathcal{D}'$ if
there is any point $x$ in $\mathbb{C}^{n+1}$ (not merely $[0,1]^{n+1}$)
such that $f(x)=0$ and $f'(x)$ is rank deficient. This set 
$\mathcal{D'}$ is an algebraic variety, i.e., the coefficients of such $f$'s are
the roots of a polynomial system.  This means that we can apply
Demmel's theorem \cite{Demmel} to conclude that the 
expected logarithm of
the condition number of a random instance is modest. (Clearly, the distance of 
a polynomial $f$ to
$\mathcal{D}$ is bounded below by the distance of $f$ to $\mathcal{D}'$.)
 We can also apply the more
recent analysis of B\"urgisser et al.\ \cite{burgisser}
to show that the `smoothed'
condition number \cite{spielman} is modest, i.e., for any polynomial $f$
(even a degenerate one), if we select a random small perturbation
of it, then the resulting polynomial is expected to have
a modest logarithmic condition number.

\begin{proof}
Let $\tilde f$ satisfy \eref{eq:hatfminusf1} and
let $e=f-\tilde f$ (i.e., a polynomial), so that
$\Vert e\Vert\le c_m\Vert f\Vert /\cond{f}$.
Choose an
$x\in[0,1]^{n+1}$.  By definition of $\cond{f}$,
$$M\cdot\min(1/\Vert f(x)\Vert,\Vert \mpinv{f'(x)}\Vert)\le\cond{f}.$$
We take two cases depending on which term achieves the min.
First, suppose $M/\Vert f(x)\Vert\le\cond{f}$.  In this case,
$\Vert f(x)-\tilde f(x)\Vert=\Vert e(x)\Vert \le \Vert e\Vert
\le c_m\Vert f\Vert/\cond{f}$. 
 Assume $c_m$ is sufficiently small so that $c_m<1/3$.
Then
\begin{eqnarray*}
\frac{1}{\Vert \tilde f(x)\Vert} &\le &\frac{1}{\Vert f(x)\Vert-\Vert f(x)-\tilde
f(x)\Vert} \\
&\le & \frac{1}{\Vert f\Vert/\cond{f}-c_m\Vert f\Vert/\cond{f}} \\
&\le & 1.5\cond{f}/\Vert f\Vert.
\end{eqnarray*}
In particular, $\tilde f(x)\ne 0$.

For the other case, the hypothesis is
$\Vert f\Vert\cdot \Vert \mpinv{f'(x)}\Vert \le \cond{f}$.
We recall that 
\begin{equation}
\label{eq:fromgvl}
\frac{1}{\sigma_n(f'(x))\sqrt{n}}
\le \Vert \mpinv{f'(x)}\Vert_\infty \le
\frac{\sqrt{n+1}}{\sigma_n(f'(x))}
\end{equation}
(see \cite[(2.3.11) and \S5.5.4]{gvl}),
where $\sigma_n(\cdot)$ is notation for the $n$th
singular value of a matrix.  By the equivalence of norms, \eref{eq:fromgvl} is equivalent to
\[
\frac{\delta_n}{\sigma_n(f'(x))\sqrt{n}}
\le \Vert \mpinv{f'(x)}\Vert \le
\frac{\gamma_n\sqrt{n+1}}{\sigma_n(f'(x))}
\]
for any arbitrary $p$-norm, where $\delta_n$ and $\gamma_n$ are positive constants that depend on $n$ and the choice of norm.

A second fact from numerical linear algebra is that
$$|\sigma_i(A)-\sigma_i(B)|\le \Vert A-B\Vert_2\le \sqrt{n}\Vert A-B\Vert$$
(see \cite[(2.3.11) and Cor.~8.6.2]{gvl}).  
Combining these facts together with \eref{eq:ebound} yields
\begin{eqnarray*}
\sigma_n(\tilde f'(x)) &=& \sigma_n(f'(x)-e'(x)) \\
&\ge & \sigma_n(f'(x))- \sqrt{n}\Vert e'\Vert \\
&\ge & \sigma_n(f'(x))- 2\alpha_n\sqrt{n}(n+1)\max(m_1,\ldots,m_{n+1}) \Vert e\Vert \\
&\ge & \frac{\delta_n}{\sqrt{n}\Vert\mpinv{f'(x)}\Vert} - \frac{2\alpha_n\sqrt{n}(n+1)\max(m_1,\ldots,m_{n+1})c_m\Vert f\Vert}{\cond{f}}
\\
&\ge & \frac{\delta_n \Vert f\Vert}{\sqrt{n}\cond{f}} - \frac{2\alpha_n\sqrt{n}(n+1)\max(m_1,\ldots,m_{n+1})c_m\Vert f\Vert}{\cond{f}}.
\end{eqnarray*}
If we select $c_m<\delta_n/\left(2\alpha_n n(n+1)\max(m_1,\ldots,m_{n+1}) \right)$, then we are assured in this case
that $\sigma_n(\tilde f'(x))>0$, i.e., the rank of $\tilde f'(x)$ is
$n$.

Thus, combining the cases, we have shown that for all $x\in[0,1]^{n+1}$, 
either
$\tilde f(x)\ne 0$ or the rank of $\tilde f'(x)$ is at least $n$, 
thus proving that $\tilde f$ is not degenerate.

Next, let us turn to \eref{eq:hatfminusf2}.   
In this proof, we will drop the prime from 
$c_m'$  and indeed will allow $c_m$ to denote a constant
depending on the degrees that may change from line to line.

Let $x_0\in[0,1]^{n+1}$ be the point where the max in \eref{cond_def} is
achieved, and let $k$ be the value of this max, i.e., the
value of $\cond{f}$. This means, first, that
$\Vert f(x_0)\Vert\le M/k$.  Second, it
means that either (i)
$\Vert \mpinv{f'(x_0)}\Vert\ge k/M$,  or (ii) $\rank(f'(x_0))<n$.
Let $A=f'(x_0)$. 

If case (i) is true, then by definition of the matrix
$p$-norm there exists a unit vector (in the norm under consideration) 
$r$
such that $\Vert A^T(AA^T)^{-1}r\Vert\ge k/M$.  Let $q=A^T(AA^T)^{-1}r$,
so that $\Vert q\Vert \ge k/M$ and $Aq=r$.  
Let $S=rq^T/(q^Tq)$ (a $n\times (n+1)$
matrix) and define $\hat f(x)\equiv f(x)-f(x_0)-S(x-x_0)$.  
In case (ii) (when $\rank(f'(x_0))<n$), define
$\hat f(x)\equiv f(x)-f(x_0)$, i.e., take $S=0$, and we do not need
$r$ and $q$.

For either case,
we now must establish that $\hat f$ is a degenerate instance
and that inequality \eref{eq:hatfminusf2} holds.

First, let us establish the degeneracy of $\hat f$.  Clearly $x_0$
is a root of $\hat f$.  In case (ii), $\hat f'(x_0)=f'(x_0)$, a matrix
whose rank is less than $n$.
In case (i),
$\hat f'(x_0)=f'(x_0)-S=A-rq^T/(q^Tq)$.
We claim that $\hat f'(x_0)$ has at least two independent vectors in its null space,
which implies that its rank is at most $n-1$.  Observe first that $A$, as 
an $n \times (n+1)$ matrix, must have a nonzero vector $w$ in its null space.
Then $w$ is also in the null space of $\hat f'(x_0)$ since $Aw=0$
and $q^Tw=r^T(AA^T)^{-1}Aw=0$.  Also, $\hat f'(x_0)$ has $q$ in
its null space as $\hat f'(x_0)q$ evaluates to $Aq-r$.  Finally, $q$ and
$w$ are independent since $Aw=0$ whereas $Aq=r$, which
is not zero.  This concludes the argument that $\hat f$ is
degenerate.  

Next, consider $\Vert\hat f-f\Vert$ appearing in
\eref{eq:hatfminusf2}.  
Let us introduce the following notation for this
argument: $\tau$ denotes the polynomial $x\mapsto x-x_0$ and
$\kappa$ denotes the constant real-valued polynomial 1. 
With these definitions, $\hat f-f=f(x_0)\kappa + S\tau$.
Note that $\Vert \tau\Vert \le c_m$ since $\Vert x_0\Vert\le 1$.
Also, $\Vert \kappa \Vert =1$.
We have the following chain of inequalities,
which applies to case (i).
 The first line involves
norms of polynomials, whereas the remaining lines are norms of vectors
and matrices.  
\begin{eqnarray}
\Vert\hat f-f\Vert & = &
\Vert f(x_0)\kappa + S\circ \tau\Vert \nonumber \\
&\le & c_{m} \left(\Vert f(x_0)\Vert +
\Vert S\Vert\right) \nonumber\\
&\le & c_{m} \left(\frac{M}{k} + 
\Vert rq^T/(q^Tq)\Vert \right)\nonumber \\
&\le& c_{m} \left(\frac{M}{k} +  \Vert r\Vert \cdot \Vert q^T/(q^Tq)\Vert \right)\nonumber\\
&\le &c_{m} \left(\frac{M}{k} +  \Vert r\Vert \cdot 1/\Vert q\Vert \right)\nonumber\\
&\le &c_{m} \left(\frac{M}{k} +  \Vert r\Vert \cdot M/k \right)
\nonumber \\
&\le &c_{m} \frac{M}{k}. \label{eq:hatfminusf_numer}
\end{eqnarray}
The first line follows from the definition of $\hat f$. The second uses the
fact that $\Vert \kappa\Vert$ and $\Vert \tau\Vert$
are bounded by constants.  The third line uses the inequality established
earlier that $\Vert f(x_0)\Vert\le M/k$ and also
expands the definition of $S$.  The fifth line uses the fact that
$q^Tq=\Vert q\Vert^2_2$, and the 2-norm and the $p$-norm under
consideration
are related by constants  depending on $m_1, \ldots, m_{n+1}$.  
The last line uses the assumption that $r$
is a unit vector.

For case (ii), \eref{eq:hatfminusf_numer} also holds since $S=0$, so
the result is already established by the third line in the above chain
of inequalities.
Since the denominator occurring in the left-hand side of
\eref{eq:hatfminusf2} is $M$,
and
recalling that $k=\cond{f}$, we see that
\eref{eq:hatfminusf_numer} proves the theorem.
\end{proof}

Although the condition number defined by \eref{cond_def} is scale invariant
(i.e., $\cond{f}\equiv\cond{cf}$ for $c\ne 0$), it is not
affinely invariant.
In other words, if $A$ is a nonsingular $n \times n$ matrix, then
in general $\cond{f}\ne\cond{Af}$.  On the other hand, our algorithm
is affinely invariant as we shall see in Section~\ref{sec:algo}.  Therefore,
we can define a new condition number that is indeed affinely invariant,
which is as follows:
\begin{equation}
\icond{f}=\inf_{A\in GL(n,\mathbb{R})}\cond{Af}.
\label{eq:invarcond}
\end{equation}
Here, $GL(n,\mathbb{R})$ denotes the set of all $n\times n$ nonsingular
matrices.
Obviously, \eref{eq:invarcond} is affinely invariant, and also,
it is clear that for any instance $f$ of SDPS, $\icond{f}\le\cond{f}$.
Furthermore, if we are able to show that our affinely invariant algorithm
has running time bounded in terms of $\cond{f}$, then it will follow
automatically that it is also bounded in terms of $\icond{f}$.

The difficulty with \eref{eq:invarcond} is that there is no obvious
way to compute this quantity other than the exhaustive
method of trying out all choices of $A$
in a dense grid lying in $GL(n,\mathbb{R})$.  (If the matrix 2-norm is
used, then it suffices to try a dense sampling of upper triangular matrices,
a smaller search space,
since the $Q$ in a $QR$ factorization of $A$ does not affect the norm.)
Unless a better method
can be found,  
definition \eref{eq:invarcond}
would be useful in practice
mainly in cases where there is {\em a priori} information
about a linear transformation that improves the condition number.

\section{Performance of other algorithms for $n=3$ case on well conditioned instances}
\label{section_other}

Recall that SSI is a special case of SDPS with $n=3$.  Due to the 
abundance of SSI algorithms in literature, we compare our algorithm to well-known SSI algorithms.  Many previously published SSI algorithms work
well in practice and are widely used in computer-aided geometric design software.  Nonetheless, we
suspect that most of these algorithms can behave nonrobustly in the sense
that, given a well conditioned problem instance, they can sometimes internally
generate an arbitrarily ill conditioned subproblem that they then
must solve.  If an algorithm
is capable of this behavior, then it is not possible to bound
its running time in terms of the instance's condition number as we shall do
for our algorithm.
Indeed, as far as we know, there is no \emph{a priori} running time upper bound for any SSI algorithm in the literature.
In this section, we consider how two well-known SSI
algorithms can generate bad subproblems given good problem instances.

\subsection{Marching methods based on collinear normal points}
\label{section_sederberg}
Collinear normal points are points on the two surfaces whose normals
are collinear.  Marching methods based on collinear normal points
split the parametric domains in at least one direction at these
collinear normal points.  The consequence is all solution curves
have one point on the boundaries of the resulting subdomains provided
that the dot product of any two normal vectors of either surfaces is
never zero.  These points are located by a curve/surface intersection
algorithm and used as starting points for the marching step.

Consider applying collinear normal points marching methods to find
intersections between the two B\'{e}zier surfaces $p$ and $q$ whose
control points are defined in Table \ref{table_sed1} and Table
\ref{table_sed2}, respectively.  This problem is equivalent to solving SDPS with control points
\[
b_{i_1,i_2,i_3,i_4} = a_{i_1,i_2} - a'_{i_3,i_4}, 
\]
where $a_{i_1,i_2}$'s $(i_j = 0, \ldots, m_j)$ are control points of $p$ and $a'_{i_3,i_4}$'s $(i_j = 0, \ldots, m_j)$ are control points of $q$.
The condition number of this instance
is $423.4$, which is reasonably well-conditioned.  The instance
is also intuitively well-conditioned as there are neither complicated
nor almost singular intersections.  The two surfaces and their
intersection in object space are shown in Figure \ref{fig_sed}.

\begin{table}
\begin{center}
{\scriptsize{
\begin{tabular}{|c|c|c|c|}
\hline 
$a_{i_1,i_2}$ & $0$ &$1$&$2$\\
\hline
$0$&$(.155,.055,.002524)^T$&$(.155,.555,.003592)^T$ & $(.155,1.055,-.008142)^T$ \\
$1$&$(.655,.055,.005414)^T$ &$(.655,.555,-.01454)^T$&$(.655,1.055,.005146)^T$ \\
$2$&$(1.155,.055,-.01745)^T$ &$(1.155,.555,-.02108)^T$& $(1.155,1.055,.01718)^T$   \\
\hline
\end{tabular}
}}\end{center}
\caption{The control points of B\'{e}zier surfaces $p$.  The entry in the $i_1$th row and the $i_2$th column is the control point $a_{i_1,i_2}$ of $p$.\label{table_sed1}}
\end{table}
\begin{table}
\begin{center}
{\scriptsize{
\begin{tabular}{|c|c|c|c|c|}
\hline 
$a'_{i_3,i_4}$ & $0$ &$1$&$2$&$3$\\
\hline
$0$&$(.1768,.1295,.02303)^T$&$(.1767,.3465,.06306)^T$&$(.1767,.6467,-.0801)^T$&$(.1768,1.437,-.01946)^T$\\
$1$&$(.4081,.1287,.04006)^T$&$(.4081,.3434,.0948)^T$&$(.4081,.6384,-.08438)^T$&$(.4081,1.444,.004442)^T$\\
$2$&$(.7515,.1249,-.07777)^T$&$(.7515,.3294,-.1071)^T$&$(.7515,.6008,.05544)^T$&$(.7515,1.477,-.01756)^T$\\
$3$&$(1.6068,.1078,-.004634)^T$&$(1.6068,.2651,.01833)^T$&$(1.6068,.4288,-.01627)^T$&$(1.6069,1.628,.2468)^T$\\
\hline
\end{tabular}
}}
\end{center}
\caption{The control points of B\'{e}zier surfaces $q$.  The entry in the $i_3$th row and the $i_4$th column is the control point $a'_{i_3,i_4}$ of $q$.\label{table_sed2}}
\end{table}

\begin{figure}
\centering
\includegraphics[width=.9\textwidth]{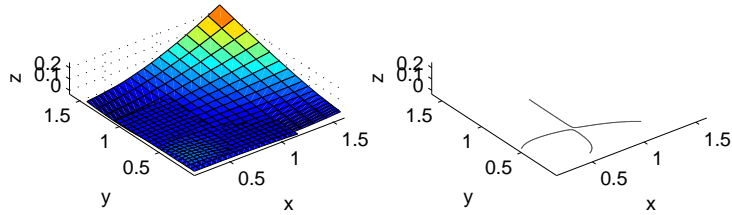}
\caption{An example of a well-conditioned instance of SSI 
problem where collinear normal point methods and Koparkar's algorithm 
require long computation time to solve.}
\label{fig_sed}
\end{figure}

Figure \ref{fig_sed_para} shows a pair of collinear normal points in
parametric space and the subdivision of domain at these points, as
well as the intersection between the surfaces (Other pairs of
collinear normal points are not shown).  Parts of the intersection
curves lie almost parallel to the boundaries of the created
subdomains.  In other words, the cuts made by the algorithm happen to
intersect degenerately or nearly degenerately with the problem data.
Thus, after the algorithm has made a subdivision of the domain
based on the collinear normals,
it must now recursively solve
arbitrarily ill conditioned subproblems since the
intersection point between the surface and the boundary curve is
almost singular. The running time of algorithms for finding these
intersection points typically depends on their conditioning
(see, e.g., the RIA algorithm of \cite{patrikalakis}).
Thus, the running time of the algorithm cannot be
bounded in terms of the condition number of the input.  
Furthermore, ill-conditioning of the
subproblems may cause unexpected inaccuracy of the solution of the
original well-conditioned instance.

\begin{figure}
\centering
\includegraphics[width=.9\textwidth]{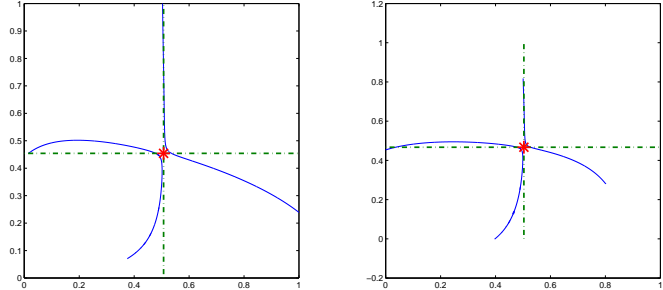}
\caption{Intersection of the surfaces in Figure \ref{fig_sed} in parametric space and splitting of the domains at a pair of collinear normal points showing that marching methods based on collinear normal points create ill-conditioned subproblem instances from a well-conditioned original instance.  Solid lines are the intersections. Asterisks are the pair of collinear normal points.  Dotted lines are the splitting lines at those points. To the left is the parametric domain of $p$.  To the right is that of $q$.}
\label{fig_sed_para}
\end{figure}

\subsection{Koparkar's algorithm}
\label{section_koparkar_cond}

Koparkar's algorithm uses a test based on the contraction mapping theorem
to determine if, in a given domain, a Newton-like method converges or
the two surfaces do not intersect at all.  If the Newton-like method
is guaranteed to converge, the method is used to find part of the
solution curves inside the domain.  If the two surfaces are known 
not to intersect, the domain is discarded.  Otherwise, each surface is
subdivided by splitting their parametric domains into four, 
and the test is repeated on the created subdomains.
This process continues until the entire domain is examined.

The test in Koparkar's algorithm requires the ability to evaluate
ranges of functions, which is typically accomplished by variety of
interval arithmetics.  None of these techniques can give the exact
ranges, however; they yield supersets of the ranges.  For this reason,
the convergence test is likely to fail when part of the
solution lies very close or directly on the border of a subcube in
both $x_1 x_2$-space and $x_3 x_4$-space at the same time, which is
not necessarily on or near the border of the original domain
$[0,1]^4$.  The same problem instance discussed in Section
\ref{section_sederberg}, which is shown in Figure \ref{fig_sed}, has
such problem.  The domain $[0,1]^4$ does not pass the test, and the
domain is subdivided at the midpoints as shown in Figure
\ref{fig_koparkar}.  The subdomains now have solutions directly on a
boundary.  Koparkar's algorithm needs to subdivide these subdomains to
very small ones before the convergence test is satisfied.  This
example demonstrates that Koparkar's algorithm is inefficient at
solving certain well-conditioned instances.  The algorithm may solve
other instances with higher condition numbers but without any parts of
the intersections near any boundaries of the subdivided subdomains
faster than this well-conditioned instance.

\begin{figure}
\centering
\includegraphics[width=.9\textwidth]{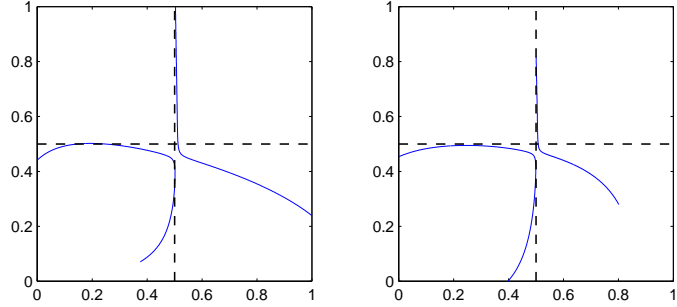}
\caption{Intersection of the surfaces in Figure \ref{fig_sed} in parametric space and subdomains created after the original domain $[0,1]^4$ fails Koparkar's test.  Due to part of the intersections lying on a boundary of a subdomain, Koparkar's algorithm requires many subdivisions before those parts of the intersections can be located.  Solid lines are the intersections.  Dashed lines are the subdivision lines. To the left is the parametric domain of $p$.  To the right is that of $q$.}
\label{fig_koparkar}
\end{figure}

\section{The Kantorovich-Test Subdivision algorithm}
\label{sec:algo}

This section describes our algorithm for the SDPS problem. 
Some details are postponed to the next section.
Since we are interested in solutions of $f$ within the hypercube $[0,1]^{n+1}$, and the closed ball $\bar B(x,r)$ defined in the infinity norm is a hypercube, our algorithm uses the infinity norm for all of its norm computation.  Therefore, for
the rest of this article, the notation $\norm{\cdot}$ is used to
refer specifically to infinity norm.

During the computation, our algorithm maintains a list of
\emph{explored regions} defined as parts of the domain $[0,1]^{n+1}$
guaranteed by Kantorovich's Theorem to contain only
the solutions that have already been found.  This list is used
in addition to another test to determine whether to subdivide a
hypercube.  We define the \emph{Kantorovich test} on a hypercube
$X = \bar{B}(x^0,r)$ as the application of Kantorovich's Theorem
on the point $x^0$ to the function $h^{(ik)}(x) = \left(f(x),x_i-k
\right)^T$ for each $i=1,2,\ldots, n+1$ and any $k \in [x_i^0-r,
x_i^0+r]$.  The hypercube $[-.5,1.5]^{n+1}$ is used as the domain
$D$ in the statement of the theorem, and
$\norm{\left[\left(h^{(ik)}\right)'(x^0)\right]^{-1}h^{(ik)}(x^0)}$
is used as $\eta$. For $\omega$, we instead use $\hat{\omega} \geq
\omega$, where $\hat{\omega}$ is defined by (\ref{computedlip})
below, as the minimal $\omega$ is too expensive to compute. The
hypercube $X$ passes the Kantorovich test if there exists an $i
\in \{1,2, \ldots, n+1 \}$ such that for every $k \in [x_i^0-r, x_i^0+r]$,
$\eta \hat{\omega} \leq 1/4$ and $\bar{B}(x^0,\rho_-) \subseteq
D$.

The choice of $D$ mentioned in the previous paragraph is used in
the analysis of the algorithm in Section~\ref{section_analysis}, but
in practice a smaller $D$ such that $B(x^0,r)\subset D\subset [-.5,1.5]^{n+1}$
may be used.  The advantage of using a smaller $D$ is that the Lipschitz
constant $\hat\omega$ will be smaller, so the 
inequality $\eta\hat\omega\le 1/4$
may be satisfied
more easily (i.e., for larger hypercubes) than using the full $D$.
The disadvantage, however, is that if $D$ is chosen too small,
then
the condition $B(x^0,\rho^-)\subset D$ of the theorem may be hard
to satisfy. In particular, the choice $D=B(x^0,r)$ would be
unacceptable for this reason.

If $X$ passes the Kantorovich test, then three important consequences
follow.  First, $x^0$ is a fast starting point for
$h^{(ik)}$ for the particular $i$ that satisfies the condition of
the Kantorovich test and any $k \in [x_i^0-r, x_i^0+r]$. Second,
the segment of the solution curve of $f$ that contains the root
guaranteed by the conclusion of Kantorovich's theorem is not a loop in $x_1 x_2 \cdots x_{n+1}$-space inside $X$ (although it may be part of a loop in the
original domain).  Third, an explored region for this segment of the
solution curve can be derived. The explored region is
\begin{equation}
\label{explored_region} X_E = \left\{x : x_i^0-r \leq x_i \leq
x_i^0+r\right\} \cap D \cap \bigcap_{k \in [x_i^0-r, x_i^0+r]}
\left( \bar{B}(x^0,\rho_-^{(k)}) \cup B(x^0,\rho_+^{(k)}) \right),
\end{equation}
where $\rho_-^{(k)}$ and $\rho_+^{(k)}$ are $\rho_-$ and $\rho_+$
in the statement of Kantorovich's theorem with respect to
$h^{(ik)}$.  Observe that $X_E$ is a hyperrectangle in
$\mathbb{R}^{n+1}$ and can be stored and computed succinctly as
detailed in Section \ref{section_impkan}.  Note also that the
explored region provides an effective way to prevent the points on
different but nearby solution curves from being incorrectly
joined into the same curve.

The other test our algorithm uses is the exclusion test. For a
given hypercube $X$, let $\hat{f}_X$ be the Bernstein polynomial
that reparametrizes with $[0,1]^{n+1}$ the function defined by $f$ over
$X$.  In other words, $\hat{f}_X(q)\equiv f(\lambda(q))$, where
$\lambda(q)$ is a composition of a dilation and translation
(uniquely determined) such that $\lambda:[0,1]^{n+1}\rightarrow X$ is
bijective.  (See Section~\ref{sec:reparam} for information how to
efficiently compute $\hat {f}_X$.)
The hypercube $X$ passes the \emph{exclusion test} if the
convex hull of the control points of $\hat{f}_X$ excludes the
origin.  It is a well-known property of the Bernstein-B\'ezier representation
of a polynomial $f$ that $f(U)$
lies in the convex hull of its control
points, where $U$ represents its natural parametric domain.
Thus, if the hull of the coefficients of a polynomial system
excludes the origin, this system has no solutions in the parametric domain.
We can check whether the hull excludes the origin by solving a low-dimensional
linear programming problem.  Megiddo \cite{Megiddo} showed that low-dimensional linear
programming problems can be solved in linear time (i.e., linear in the number
of control points), although we have not used Megiddo's algorithm.
Note that some other polynomial bases also have a similar exclusion test;
refer to \cite{srijuntongsiri_basis}.

We now proceed to describe our algorithm, the
\emph{Kantorovich-Test Subdivision algorithm} or KTS in short.

\begin{flushleft}
\textbf{Algorithm KTS}:
\end{flushleft}
\begin{itemize}
\item Let $Q$ be a queue with $[0,1]^{n+1}$ as its only entry. Set $S
= \emptyset$. \item Repeat until $Q = \emptyset$
\begin{enumerate}
\item Let $X$ be the hypercube at the front of $Q$.  Remove $X$
from $Q$. \item If $X \not\subseteq X_{E'}$ for all $X_{E'} \in
S$,
\begin{itemize}
\item Perform the exclusion test on $X=\bar{B}(x^0,r)$
\item If $X$ fails the exclusion test,
\begin{enumerate}
\item Perform the Kantorovich test on $X$
\item If $X$ passes the Kantorovich test,
\begin{enumerate}
\item Perform Newton's method on $h^{(ik)}$, where $i$ is the
index that satisfies the condition of the Kantorovich test and $k
= x_i^0-r$, starting from $x^0$ to find a zero $x^*$.

\item Trace the segment of the solution curve using $x^*$ as
the starting point and going toward $x_i^0+r$ direction until the
$x_i = x_i^0+r$ boundary is reached.

\item If the newly found segment is contained in any $X_{E'} \in
S$ (i.e. the segment has been found before), discard the segment.

\item Otherwise, compute the new explored region $X_E$ according
to (\ref{explored_region}).  Set $S = S \cup \{X_E\}$.
\end{enumerate}
\item If either $X$ fails the Kantorovich test or $X$ passes the
test with $X \not\subseteq X_E$, subdivide $X$ along all $n+1$
parametric axes into $2^{n+1}$ equal smaller hypercubes. Add these hypercubes to
the end of $Q$.
\end{enumerate}
\end{itemize}
\end{enumerate}
\item Check if any two segments of solution curves overlap. If
so, remove the overlapping part from one of the segments.

\item Join any two segments sharing an endpoint into one
continuous curve. Repeat until there are no two curves sharing an
endpoint.
\end{itemize}
A few remarks are needed regarding the description of the KTS
algorithm.
\begin{itemize}
\item The subdivision in step 2.c is performed in the case that
$X$ passes the Kantorovich test but $X \not\subseteq X_E$ because,
in general, passing the Kantorovich test does not imply that there
is only one solution curve in $X$.

\item The check in step 2.b.iii is necessary since the segment detected
by the Kantorovich test may be outside of $X$.

\item For the same reason as above, certain parts of a solution curve may be traced twice and hence must be removed
from one of the segments before the segments are joined.  The
overlapping segments can be detected by checking if an endpoint of
a segment is inside an explored region of another segment.  The
segments sharing an endpoint can also be detected from explored
regions in a similar manner.  Note that there is no ambiguity in this
step because an explored region, having passed the Kantorovich test,
cannot contain more than one connected component of the solution.

\item If the Kantorovich test is not applicable for a certain
hypercube due to the Jacobian of the midpoint being singular, the
hypercube is treated as if it fails the Kantorovich test and is then subdivided by step 2.c.
\end{itemize}

One property of KTS is that it is affinely invariant.  In other
words, left-multiplying $f$ with an $n \times n$ matrix $A$ prior to
executing KTS does not change its behavior. This is the main
reason that we introduced $\icond{f}$ earlier.

%

\section{Implementation details}
\label{sec:impdetails}

The implementations of certain steps of KTS are not apparent and thus
are explained in detail in this section.  While this
section focuses only on the B\'{e}zier surface case,
all the results herein can be generalized to certain other polynomial
bases as mentioned in the introduction.

\subsection{Computation of Lipschitz constant}
\label{section_implip}

For simplicity, denote $h^{(ik)}$ as $h$ when the choice of $(ik)$
is clear from context.  The Lipschitz constant for $h'(x^0)^{-1}h'
\equiv g$, which is required for the Kantorovich test, is obtained
from an upper bound over all $x \in [-.5, 1.5]^{n+1}$ of the derivative of $g$
\[
g'(x) =
\left(\frac{\partial^2\left(h'(x^0)^{-1}h\right)_i(x)}{\partial
x_j
\partial x_k} \right),
\]
where $\left(h'(x^0)^{-1}h\right)_i(x)$ denotes the $i$th entry of
$\left(h'(x^0)^{-1}h\right)(x)$. Let $\hat{g}$ be
the Bernstein polynomial that reparametrizes with $[0,1]^{n+1}$ the
surface defined by $g$ over $[-.5, 1.5]^{n+1}$ (Refer to Section~\ref{sec:reparam}).
We have
\begin{eqnarray}
\max_{x \in [-.5, 1.5]^{n+1}} \norm{g'(x)} & = & \frac{1}{2} \cdot \max_{x \in [0,1]^{n+1}} \norm{\hat{g}'(x)} \nonumber \\
& = & \frac{1}{2} \cdot \max_{x \in [0,1]^{n+1}} \max_{\norm{y}=1} \norm{\hat{g}'(x)y} 
\nonumber \\
                            & \leq &\frac{1}{2} \cdot \max_{x \in [0,1]^{n+1}}
 \max_i \sum_{j=1}^{n+1}\sum_{k=1}^{n+1} |\hat{g}'_{ijk}(x)| 
\nonumber \\
                            & \leq & \frac{(n+1)^2}{2} \max_{i,j,k} \max_{x \in [0,1]^{n+1}} | \hat{g}'_{ijk}(x) | .\nonumber
\end{eqnarray}
Note that each entry of $\hat{g}'$ can be written as a Bernstein
polynomial efficiently because
\begin{equation}
\frac{d Z_{i,m}(t)}{dt} = m\left(
Z_{i-1,m-1}(t)-Z_{i,m-1}(t)\right),
\label{eq:bezderiv1}
\end{equation}
where $Z_{-1,m-1}(t)=Z_{m,m-1}(t)=0$, which can be used to compute
the control points of the derivatives in Bernstein basis from a
given Bernstein polynomial directly. Hence, the maximum absolute
value of the control points of $\hat{g}'_{ijk}$ when written in
Bernstein basis is an upper bound of $\max_{x \in [0,1]^{n+1}} |
\hat{g}'_{ijk}(x)|$. Let $\hat{\omega}$ denote the Lipschitz
constant computed in this manner, that is,
\begin{equation}
\label{computedlip} \hat{\omega} \equiv \frac{(n+1)^2}{2} \max_{i,j,k,i_1,i_2,\ldots,i_{n+1}} 
\Vert\hat{g}'_{ijk,i_1,i_2,\ldots,i_{n+1}}\Vert
,
\end{equation}
$\hat{g}'_{ijk,i_1,i_2,\ldots,i_{n+1}}$ are the control points
of $\hat g'_{ijk}$.

\subsection{The Kantorovich test and solution curve tracing}
\label{section_impkan}

Recall that for a hypercube $X$ to pass the Kantorovich test,
there must exist an $i \in \{1,2,\ldots,n+1\}$ satisfying $\eta \hat
\omega \leq 1/4$ and $\bar{B}(x^0,\rho_-) \subseteq D$ for
\emph{all} functions $h^{(ik)}$'s where $k \in [x^0_i-r,
x^0_i+r]$. The algorithm, however, need not explicitly check the
conditions for all values of $k$.  Notice that $\hat \omega$ and
$D$ are independent of $k$ and $\rho_-$ is an increasing function
of $\eta$.  For these reasons, KTS only needs to check the
conditions for the value of $k$ that maximizes $\eta$. Similarly,
the explored region $X_E$ can be computed solely from the
maximizer $k$.  But note also that $\eta$ is linear in $k$, which
means that the value of $k$ maximizing $\eta$ is either $x^0_i-r$
or $x^0_i+r$.

After a hypercube passes the Kantorovich test, the segment of the
solution curve detected by the test must be traced.  Since
the Kantorovich test guarantees that performing Newton's
method on $h^{(ik)}$ starting on $x^0$ converges for any $k \in
[x^0_i-r, x^0_i+r]$, we can trace the segment by
repeating Newton's method starting on $x^0$ for many different
values of $k$ to locate the points on the segment of the
solution curve.  Alternatively, we can perform Newton's method on $h^{(i,x^0_i-r)}$ starting on $x^0$ to find a point $x^1$ on the segment, use $x^1$ as the starting point for Newton's method on $h^{(i,x^0_i-r+\epsilon)}$, $\epsilon > 0$,
to find the next point $x^2$ on the segment, use $x^2$ as the starting point on 
$h^{(i,x^0_i-r+2\epsilon)}$ to find the next point on the segment, and so on.

\subsection{Reparametrization}
\label{sec:reparam}

There are two steps of KTS involving reparametrization of
polynomials in Bernstein basis, namely the exclusion test and the
computation of the Lipschitz constant for the Kantorovich test.
``Reparametrization'' in this context means the computation of new
control points that describe the same function with respect
to the new parameter domain.
Both steps require the reparametrization with $[0,1]^{n+1}$ of
Bernstein polynomials with $n+1$ variables, which is a
straightforward extension of reparametrization with $[0,1]^2$ of
bivariate Bernstein polynomials.  An example of efficient 
algorithms for reparametrizing bivariate Bernstein polynomials with $[0,1]^2$ is discussed in \cite{srijuntongsiri_basis}.  Alternatively, the polynomials can be reparametrized by two applications of the \emph{de Casteljau algorithm} \cite{patrikalakis}; one to reparametrize the right endpoints with $1$ and another to reparametrize the left endpoints with $0$.

\section{Time complexity analysis}
\label{section_analysis}

In this section, we prove a number of theorems leading to the
theorem regarding the running time of the KTS algorithm. Since
both the exclusion test and the computation of the Lipschitz
constant in the Kantorovich test use the control points in their
computations, it is useful to find the relationship between the
control points and the function values of the polynomial defined
by them. Recall that $M$ is defined as the maximum norm among
control points of $f$ and was denoted $\Vert f\Vert$ earlier.

We have already shown that $f(x)$ and $f'(x)$ are bounded
by $M$ in 
\eref{eq:fMbound} and \eref{eq:fpMbound}.
Using the same logic, we can derive
a Lipschitz bound on $f'(x)$, i.e., an upper bound on
$f''$, as follows:
\begin{equation}
\Vert f'(x)-f'(y)\Vert \le 4(n+1)\max(m_1,\ldots,m_{n+1})^2M\Vert x-y\Vert.
\label{eq:fLipMbound}
\end{equation}

The use of the Kantorovich theorem requires a Lipschitz bound
for a slightly larger region.  If we require a Lipschitz bound for
$f'(x)$ over $[-\eps,1+\eps]^{n+1}$, then we can argue based on the
deCasteljau algorithm for evaluating B\'ezier polynomials that
\begin{equation}
  \Vert f'(x)-f'(y)\Vert \le 4(n+1)(1+\eps)^{\max(m_1,\ldots,m_{n+1})}\max(m_1,\ldots,m_{n+1})^2M\Vert x-y\Vert.
\label{eq:exLip}
\end{equation}

Another useful inequality is that for any $x\in[0,1]^{n+1}$,
\begin{equation}
M\Vert \mpinv{f'(x)}\Vert \ge 1/\left(2(n+1)\max(m_1, \ldots, m_{n+1})\right).
\label{eq:Mpinvineq}
\end{equation}
Equation \eref{eq:Mpinvineq} can be established
by multiplying both sides of \eref{eq:fpMbound}
by  $\Vert \mpinv{f'(x)}\Vert$ and then using
the fact that
$\Vert f'(x)\Vert\cdot \Vert \mpinv{f'(x)}\Vert\ge
\Vert  f'(x)\mpinv{f'(x)}\Vert =\Vert I\Vert=1$ for any
matrix $p$-norm, and for the infinity-norm in particular.

Now, we establish a bound that reverses \eref{eq:fMbound},
namely,
\begin{equation}
\label{theta_eq} M \leq \theta \max_{0 \leq
x_1,x_2,\ldots,x_{n+1} \leq 1} \norm{f(x)}
\end{equation}
for any polynomial $f$, where $\theta$ is as
defined as follows.
\begin{eqnarray}
\label{theta_def} \theta = \theta(m_1,\ldots, m_{n+1}) & = &
\prod_{k=1}^{n+1}\left(\sum_{i=0}^{m_k} \prod_{i' \neq i} \frac{\max\{|m_k-i'|,|i'|
\}}{|i-i'|} \right) \nonumber \\
&  = & O\left(m_1^{m_1+1}m_2^{m_2+1}\cdots m_{n+1}^{m_{n+1}+1}\right). \nonumber
\end{eqnarray}
This value of $\theta$  is specific to the choice of the
Bernstein-B\'ezier
basis; see \cite{srijuntongsiri_basis} for a discussion of other bases.

Our proof of \eref{theta_eq} is based on establishing a similar
result for univariate polynomials as shown by the following
lemmas.
\begin{theorem}[Srijuntongsiri and Vavasis \cite{srijuntongsiri_basis}]
\label{bibound} Let $f(t)$ be a polynomial system
\[
\begin{array}{llll}
f(t) &=& \sum_{i=0}^m b_{i} Z_{i,m}(t), & 0 \leq t \leq 1,
\end{array}
\]
where $b_{i} \in \mathbb{R}^d$. The norm of the coefficients can
be bounded by
\[
\norm{b_{i}} \leq  \xi_B(m) \max_{t : 0 \leq t \leq 1}
\norm{f(t)},
\]
where
\[
\xi_B(m) = \sum_{i=0}^m \prod_{j=0,1,\ldots,i-1,i+1,\ldots,m}
\frac{\max\{|m-j|,|j|\}}{|i-j|} = O(m^{m+1}).
\]
\end{theorem}

\begin{lemma}
\label{theta_to_gen} 
Let $l$ and $h$ be constants satisfying $l
< h$.  Let $n$ be a given positive integer.  Suppose there exists a function $\xi(m)$ satisfying
\begin{equation}
\label{eq:bound_1d}
\norm{a_{i}} \leq \xi(m) \max_{l \leq t \leq h} \norm{g(t)}
\end{equation}
for any $a_i \in \mathbb{R}^d$ $(i = 0,1,\ldots,m)$ and any univariate polynomial
$g(t) = \sum_{i=0}^m a_i \phi_i(t)$, where $\phi_i(t)$ denotes the
polynomial basis.  Suppose also there exists a function $\zeta(m_1, m_2, \ldots, m_n)$ satisfying
\begin{equation}
\label{eq:bound_nd}
\norm{\bar a_{i_1,i_2,\ldots,i_n}} \leq \zeta(m_1, m_2, \ldots, m_n) \max_{l \leq x_1, x_2, \ldots, x_n \leq h} \norm{\bar g(x_1,x_2,\ldots,x_n)}
\end{equation}
for any $\bar a_{i_1,i_2,\ldots,i_n} \in \mathbb{R}^d$ $(i_j = 0,1,\ldots,m_j)$ and any polynomial in $n$ variables 
\[
\bar g(x_1,x_1,\ldots,x_n) = \sum_{i_1=0}^{m_1}\sum_{i_2=0}^{m_2}\cdots \sum_{i_n=0}^{m_n} \bar a_{i_1,i_2,\ldots,i_n} \phi_{i_1}(x_1) \phi_{i_2}(x_2)\cdots \phi_{i_n}(x_n).
\]
Then
\begin{equation}
\norm{b_{i_1,\ldots, i_{n+1}}} \leq \zeta(m_1, m_2, \ldots, m_n)\xi(m_{n+1}) \max_{l \leq x_1, x_2, \ldots, x_{n+1} \leq h}
\norm{f(x_1, x_2, \ldots, x_{n+1})}
\end{equation}
for any $b_{i_1,\ldots,i_{n+1}} \in \mathbb{R}^d$ $(i_j = 0,1,\ldots,m_j)$, where
$f$ is the 
polynomial in $n+1$ variables defined by 
\begin{equation}
\label{definef}
f(x_1,x_2,\ldots,x_{n+1}) = 
\sum_{i_1 = 0}^{m_1}\cdots \sum_{i_{n+1}=0}^{m_{n+1}} b_{i_1,  \cdots, i_{n+1}} \phi_{i_1}(x_1) \cdots \phi_{i_{n+1}}(x_{n+1}).
\end{equation}
\end{lemma}
\begin{proof}
Let $f(x_1,x_2,\ldots,x_{n+1})$ be an arbitrary polynomial defined as in (\ref{definef}).  For any $n$-tuple $I' = (i'_1, i'_2, \ldots, i'_n)$ $(i'_j = 0, 1, \ldots, m_j)$, define 
\[
g_{I'}(x_{n+1}) = \sum_{i_{n+1}=0}^{m_{n+1}} b_{i'_1, i'_2, \ldots, i'_n, i_{n+1}} \phi_{i_{n+1}}(x_{n+1}).
\]
Applying (\ref{eq:bound_1d}) to $g_{I'}$ yields
\begin{equation}
\label{eq:g1}
\norm{ b_{i'_1, i'_2, \ldots, i'_n, i_{n+1}}} \leq \xi(m_{n+1}) \max_{l \leq x_{n+1} \leq h} \norm { g_{I'}(x_{n+1})}
\end{equation}
for any $b_{i'_1, i'_2, \ldots, i'_n, i_{n+1}}$ $(i_{n+1}= 0, 1, \ldots, m_{n+1})$.  Let $x^*_{I'}$ be a point where $\max_{l \leq x_{n+1} \leq h} \norm { g_{I'}(x_{n+1})}$ is achieved.  Define 
\[
k_{I'}(x_1, x_2, \ldots, x_n) = \sum_{i_1=0}^{m_1} \sum_{i_2 = 0}^{m_2} \cdots \sum_{i_n=0}^{m_n}\sum_{i_{n+1}=0}^{m_{n+1}} b_{i_1, i_2, \ldots, i_{n+1}} \phi_{i_1}(x_1) \phi_{i_2}(x_2) \cdots \phi_{i_n}(x_n) \phi_{i_{n+1}}(x^*_{I'}).
\]
Applying (\ref{eq:bound_nd}) to $k_{I'}$ yields
\begin{equation}
\label{eq:g2}
\norm{\sum_{i_{n+1}=0}^{m_{n+1}} b_{i_1, \ldots, i_{n+1}} \phi_{i_{n+1}}(x^*_{I'})} \leq \zeta(m_1, m_2,\ldots,m_n) \max_{l \leq x_1, \ldots, x_n \leq h} \norm{ k_{I'}(x_1, \ldots, x_n)}
\end{equation}
for any $(i_1, i_2, \ldots, i_n)$ ($i_j = 0, 1, \ldots, m_j$).  Consequently, by combining (\ref{eq:g1}) and (\ref{eq:g2}),
\begin{eqnarray*}
\norm{ b_{i'_1, i'_2, \ldots, i'_n, i_{n+1}}} &\leq & \xi(m_{n+1}) \norm { g_{I'}(x^*_{I'})} \\
& = & \xi(m_{n+1})\norm{\sum_{i_{n+1}=0}^{m_{n+1}} b_{i_1, \ldots, i_{n+1}} \phi_{i_{n+1}}(x^*_{I'})} \\
& \leq & \xi(m_{n+1})\zeta(m_1, m_2,\ldots,m_n) \max_{l \leq x_1, \ldots, x_n \leq h} \norm{ k_{I'}(x_1, \ldots, x_n)} \\
& \leq & \xi(m_{n+1})\zeta(m_1, m_2, \ldots, m_n) \max_{l \leq x_1, x_2, \ldots, x_{n+1} \leq h}
\norm{f(x_1, x_2, \ldots, x_{n+1})}.
\end{eqnarray*}
\end{proof}

Hence, (\ref{theta_eq}) holds by induction using Theorem \ref{bibound} as the basis and Lemma \ref{theta_to_gen} as the inductive step.

Recall that the Lipschitz constant $\hat{\omega}$ given by
(\ref{computedlip}) is not the smallest Lipschitz constant for
$\inv{h'(x^0)}h$ over $D=[-.5,1.5]^{n+1}$. However, we
can show that $\hat{\omega} \leq \left((n+1)^2/2\right)\theta\omega$, where $\omega$
denotes the smallest Lipschitz constant for $\inv{h'(x^0)}h$ over
$D$. Since $\hat{\omega}$ is computed from the absolute values of
the control points of $\hat{g}'_{ijk}(x)$, by (\ref{theta_eq}),
\begin{eqnarray}
\hat{\omega} & \leq & \frac{(n+1)^2}{2} \theta \max_{i,j,k}\max_{x \in [0,1]^{n+1}} \left|\hat{g}'_{ijk}(x) \right| \nonumber \\
& = & \frac{(n+1)^2}{2}\theta \max_{i,j,k}\max_{x \in D} \left|g'_{ijk}(x) \right| 
 \nonumber \\
& \leq & \frac{(n+1)^2}{2}\theta \max_{x \in D} \norm{g'(x)} = \frac{(n+1)^2}{2}\theta\omega.
\label{hodef}
\end{eqnarray}
With this bound on $\hat{\omega}$, we can now analyze the behavior of the Kantorovich test.

The following is the main technical result of this article.
\begin{theorem}
\label{thm2} Let $f(x)=f(x_1,x_2,\ldots,x_{n+1})$ be a Bernstein polynomial system in $n$ dimensions such that $\cond{f}<\infty$.
Let $x^0$ be a point in $[0,1]^{n+1}$.
Let $r > 0$ be such that $\bar{B}(x^0,r) \subseteq [0,1]^{n+1}$. If
\begin{equation}
\label{deltahatcond}
r < \left(\frac{c_0}{\cond{f}}\right)^2,
\end{equation}
where
\begin{equation}
c_0=\frac{1}{64(n+1)^3\theta\cdot 1.5^{\max(m_1,m_2,\ldots,m_{n+1})}\max(m_1,m_2,\ldots,m_{n+1})^2},
\label{eq:k0def}
\end{equation}
then either
\begin{enumerate}
\item The hypercube $\bar{B}(x^0,r)$ passes the Kantorovich test
and the associated explored region $X_E$ contains $X$, or 
\item
The hypercube $\bar{B}(x^0,r)$ passes the exclusion test.
\end{enumerate}
\end{theorem}
\begin{proof}
Let $X$ denote the hypercube $\bar{B}(x^0,r)$. 
Let us introduce additional notation for frequently used quantities.
Let $\phi$ stand for $\Vert f(x^0)\Vert$, and
let $\pi$ 
stand for $\Vert \mpinv{f'(x^0)}\Vert$.
The proof is
divided into two cases by the relative value of the two
terms in the definition of condition number \eref{cond_def} 
evaluated at $x^0$.
Let 
\begin{equation}
c_1=\sqrt{64(n+1)^3\theta\cdot (1.5)^{\max(m_1,m_2,\ldots,m_{n+1})}\max(m_1,m_2,\ldots,m_{n+1})^2}.
\label{eq:k1def}
\end{equation}

\begin{flushleft} \textbf{Case 1: } $c_1M\pi \le \sqrt{M/\phi}$.
\end{flushleft}
Note that this case encompasses the possibility that $\phi=0$, i.e.,
that $x^0$ lies on a solution curve.  On the other hand, this
case requires $\pi<\infty$, i.e., $\rank(f'(x^0))=n$.
Note that $M\pi$ is the second term of \eref{cond_def} evaluated
at $x^0$ while $\sqrt{M/\phi}$ is the first term.  Therefore,
$\cond{f}\ge \min(M\pi,M/\phi)\ge \min(M\pi,c_1^2 M^2\pi^2).$
Since $c_1M\pi\ge 1$ and $c_1^2M\pi\ge 1$
(combine \eref{eq:Mpinvineq} and \eref{eq:k1def}), it follows
that the second term dominates the first, hence 
\begin{equation}
\cond{f}\ge M\pi.
\label{eq:case1condf}
\end{equation}
By the hypothesis for this case, $c_1^2 M^2\pi^2\le M/\phi$, i.e.,
\begin{equation}
\phi\le 1/(c_1^2M\pi^2).  
\label{eq:case1phibd}
\end{equation}

Let $v(x^0)$ be the unit-length null vector of $f'(x^0)$, i.e.,
$f'(x^0)v(x^0) = 0$, $\norm{v(x^0)}=1$.  Let $i$ be such that
$|v_i(x^0)| = 1$. Define 
$$h(x) = \left(\begin{array}{cc}f(x)\\ x_i - k\end{array}\right)$$
for an arbitrary $k\in[
x^0_i - r, x^0_i + r]$. By using the facts that
$\norm{v(x^0)}=1$, $|v_i(x^0)| = 1$, and
\[ h'(x^0)^{-1} = \left(\begin{array}{c}
f'(x^0) \\ e_i^T \end{array} \right)^{-1} = \left(
\left(I-\frac{v(x^0)e_i^T}{v(x^0)^T e_i}\right)\mpinv{f'(x^0)}
,\frac{v(x^0)}{v(x^0)^T e_i} \right),
\]
where $e_i$ denotes the $i$th column of the identity matrix, it is
seen that
\begin{equation}
\label{eta_ineq} \eta \equiv \norm{h'(x^0)^{-1}h(x^0)} \leq
2\left( \norm{\mpinv{f'(x^0)}f(x^0)}  + r \right),
\end{equation}
for any $k \in [x^0_i-r,x^0_i+r]$. 
(Note that in the infinity norm, $\Vert I-ve_i^T/(v^Te_i)\Vert\le 2$, 
where $i$ is the index such that $\Vert v\Vert=|v^Te_i|$.)
The first parenthesized
term on the right-hand side of \eref{eta_ineq} is clearly bounded
by $\pi\phi$, which in turn is bounded by $1/(c_1^2M\pi)$ by
\eref{eq:case1phibd}.

It follows from \eref{eq:case1condf}, \eref{eq:Mpinvineq}, and
\eref{eq:k0def} that $c_0/\cond{f}\le 1$, hence
$(c_0/\cond{f})^2\le (c_0/\cond{f})$.  Combining this with
\eref{deltahatcond} 
and \eref{eq:case1condf}
yields
\begin{equation}
r\le c_0/\cond{f}\le c_0/(M\pi).
\label{eq:rbdcase1}
\end{equation}
Thus, proceeding from \eref{eta_ineq},
\begin{equation}
\eta\le \frac{2(c_0+1/c_1^2)}{M\pi}.
\label{eq:etabound}
\end{equation}

Note that $4(n+1)(1.5)^{\max(m_1,\ldots,m_{n+1})}\max(m_1,\ldots,m_{n+1})^2M$ is a Lipschitz
constant of $f'$ on $[-5,1.5]^{n+1}$ by \eref{eq:exLip}.  Consequently, $4(n+1)(1.5)^{\max(m_1,\ldots,m_{n+1})}\cdot\linebreak\max(m_1,\ldots,m_{n+1})^2M\pi$ is a Lipschitz
constant of  $\mpinv{f'(x^0)}f'$ on $[-5,1.5]^{n+1}$. From this bound,
we derive a Lipschitz constant
$\omega$ for
 $h'(x^0)^{-1}h'$ over $[-.5,1.5]^{n+1}$ as follows.
\begin{eqnarray}
\omega & = & \max_{y,z\in D;y\ne z}
\frac{\norm{h'(x^0)^{-1}\left(h'(y)-h'(z)
\right)} }{ \norm{y-z} } \nonumber \\
& = & \max_{y,z\in D;y\ne z}\frac{1}{\norm{y-z}} \cdot
\norm{\left(I-\frac{v(x^0)e_i^T}{v(x^0)^T
e_i}\right)\mpinv{f'(x^0)}\left(f'(y)-f'(z)\right) } \nonumber \\
& \leq &\max_{y,z\in D;y\ne z} \frac{2\norm{\mpinv{f'(x^0)}\left(f'(y)-f'(z)\right)
}}{\norm{y-z}} \nonumber
\\
& \leq &2\pi\cdot \max_{y,z\in D;y\ne z} \frac{\Vert f'(y)-f'(z)\Vert}
{\norm{y-z}} \nonumber
\\
& \leq &
8(n+1)(1.5)^{\max(m_1,\ldots,m_{n+1})}\max(m_1,\ldots,m_{n+1})^2M\pi. \nonumber
\end{eqnarray}
For the last line, we used
\eref{eq:exLip}.
But from \eref{hodef},
\begin{equation}
\hat \omega \leq \frac{(n+1)^2}{2} \theta \omega 
\leq 4(n+1)^3\theta(1.5)^{\max(m_1,\ldots,m_{n+1})}\max(m_1,\ldots,m_{n+1})^2M\pi
\label{eq:omegahat}
\end{equation}
where $\omega$ is the Lipschitz constant for $h'(x^0)^{-1}h'$ over
$[-.5,1.5]^{n+1}$
and $\hat \omega$ is the \emph{computed} Lipschitz constant for
$h'(x^0)^{-1}h'$ over $D$ as defined in (\ref{computedlip}).

Combining \eref{eq:etabound} with
\eref{eq:omegahat} yields
\begin{equation}
\label{pass_first} \eta \hat \omega < 8(n+1)^3\theta(1.5)^{\max(m_1,\ldots,m_{n+1})}
\max(m_1,\ldots,m_{n+1})^2(c_0+1/c_1^2).
\end{equation}
By choice of $c_0$ and $c_1$ in \eref{eq:k0def} and \eref{eq:k1def}
respectively, we see that $\eta\hat\omega <1/4$
for any $k \in [x^0_i-r,x^0_i+r]$, which is one of the conditions
for $X$ to pass the Kantorovich test.

For the other condition, note that $\sqrt{1-2h} \geq 1-2h$ for $0
\leq h \leq 1/2$. Therefore,
\begin{eqnarray}
\rho_- & = & \frac{1-\sqrt{1-2\eta\hat\omega}}{\hat \omega}
\nonumber \\
& \leq & 2\eta \nonumber \\
&\le & \frac{4(c_0+1/c_1^2)}{M\pi}\qquad\mbox{(by \eref{eq:etabound})} \nonumber\\
&\le & 8(n+1)(c_0+1/c_1^2)\max(m_1,\ldots, m_{n+1})\qquad\mbox{(by \eref{eq:Mpinvineq})}.
\label{rho_to_gamma}
\end{eqnarray}
By choice of $c_0$ and $c_1$, we conclude that $\rho_-<1/2$ and
therefore $X\subset [-.5,1.5]^{n+1}$, the domain for which $\omega$ is
a Lipschitz constant.
This proves that $X$ satisfies the Kantorovich conditions.

Finally, the associated explored region $X_E$ contains $X$ because
\begin{eqnarray*}
\rho_+ & = & \frac{1+\sqrt{1-2\eta\hat\omega}}{\hat \omega}
\nonumber \\
& \geq & \frac{1}{\hat \omega} \nonumber \\
& \geq & \frac{1}
{4(n+1)^3\theta(1.5)^{\max(m_1,\ldots,m_{n+1})}\max(m_1,\ldots,m_{n+1})^2M\pi}
\qquad\mbox{(by \eref{eq:omegahat})} \\
&\ge & r, \nonumber
\end{eqnarray*}
for any $k \in [x^0_i-r,x^0_i+r]$.
The last line follows from \eref{eq:k0def} and \eref{eq:rbdcase1}.

\begin{flushleft} \textbf{Case 2: } $c_1M\pi \ge \sqrt{M/\phi}$.
\end{flushleft}
Note that this case encompasses the possibility that $\pi=\infty$,
i.e., $\rank(f'(x^0))<n$.  On the other hand, this case requires
$\phi>0$, i.e., $x^0$ is not a root.
Since $\cond{f}\ge\min(M\pi,M/\phi)
\ge\min(\sqrt{M/\phi}/c_1,M/\phi)$, we conclude
\begin{equation}
\cond{f}\ge \sqrt{\frac{M}{c_1^2\phi}},
\label{eq:case2condbd}
\end{equation}
which implies by \eref{deltahatcond} that
\begin{equation}
r \le \frac{c_0^2c_1^2\phi}{M}.
\label{eq:case2rbd}
\end{equation}
Select an arbitrary $x\in B(x^0,r)$.  
We now derive a bound on $f(x)-f(x^0)$ by
applying the fundamental theorem
of calculus.
\begin{eqnarray*}
\Vert f(x)-f(x^0)\Vert &=&
\left\Vert\int_0^1 f'(x^0+t(x-x^0))(x-x^0)\,dt \right\Vert\\
&\le& \int_0^1 \Vert f'(x^0+t(x-x^0))\Vert\,dt \cdot \Vert x-x^0\Vert\\
&\le & r\int_0^1 
2(n+1)\max(m_1,\ldots,m_{n+1})M\,dt \qquad\mbox{(by \eref{eq:fpMbound})}\\
&=&
2r(n+1)\max(m_1,\ldots,m_{n+1})M \\
&\le &
2(n+1)c_0^2c_1^2\max(m_1,\ldots,m_{n+1})\phi \qquad\mbox{(by \eref{eq:case2rbd})}
\\
&\le&
\frac{\phi}{2\theta} \qquad\mbox{(by \eref{eq:k0def} and \eref{eq:k1def}).}
\end{eqnarray*}
Thus, by definition of $\phi$,
\begin{equation}
\label{before_define_fhat} \theta\cdot\norm{f(x)-f(x^0)} < \norm{f(x^0)}.
\end{equation}
Define $\hat{f}(\hat{x})$ such that
\begin{eqnarray}
\hat{f}(\hat{x}_1,\hat{x}_2,\ldots, \hat{x}_{n+1}) &=& f(
2r\hat{x}_1+x^0_1-r, 2r\hat{x}_2+x^0_2-r,\ldots \nonumber
\\ && \hspace{11 pt} 2r\hat{x}_{n+1}+x^0_{n+1}-r). \label{rescale}
\end{eqnarray}
In other word, $\hat{f}$ is a polynomial that
reparametrizes with $[0,1]^{n+1}$ the surface defined by $f$ over $X$.
In terms of $\hat{f}$, (\ref{before_define_fhat}) is equivalent to
\[
\theta\cdot\norm{\hat{f}(\hat{x})-\hat{f}(\hat{x}^0)} <
\norm{\hat{f}(\hat{x}^0)}
\]
for an arbitrary $\hat{x} \in [0,1]^{n+1}$, where $\hat{x}$ is the rescaled
$x$, and $\hat{x}^0$ is the rescaled $x^0$ according to
(\ref{rescale}). In particular,
\begin{equation}
\label{beforebi} \theta\cdot\max_{\hat{x} \in
[0,1]^{n+1}}\norm{\hat{f}(\hat{x})-\hat{f}(\hat{x}^0)} <
\norm{\hat{f}(\hat{x}^0)}.
\end{equation}
Let $g(\hat{x}) \equiv \hat f(\hat{x})-\hat f(\hat{x}^0)$. 
By (\ref{theta_eq}),
\begin{equation}
\label{fromlemma} \norm{c_{i_1,\ldots,i_{n+1}}} \leq \theta\cdot\max_{\hat{x}
\in [0,1]^{n+1}}\norm{g(\hat{x})},
\end{equation}
for any control point $c_{i_1,\ldots,i_{n+1}}$ of $g(\hat{x})$, which is
equivalent to
\begin{equation}
\label{shifted_theta} \norm{a_{i_1,\ldots,i_{n+1}}-\hat f(\hat{x}^0)} \leq
\theta\cdot\max_{\hat{x} \in
[0,1]^{n+1}}\norm{\hat f (\hat{x})-\hat f(\hat{x}^0))},
\end{equation}
for any control point $a_{i_1,\ldots,i_{n+1}}$ of $\hat f(\hat{x})$ (since a constant
additive term to a polynomial corresponds to a translation of all of
its control points). Substituting
(\ref{shifted_theta}) into the left-hand side of (\ref{beforebi})
yields
\begin{equation}
\label{convexnozero} \norm{a_{i_1,\ldots,i_{n+1}}-\hat f(\hat{x}^0)} <
\norm{\hat f(\hat{x}^0)},
\end{equation}
which implies that the convex hull of the control points of
$\hat f(\hat{x})$ does not contain the origin. 
Therefore, $X$ passes the exclusion test.
\end{proof}

\section{Computational results}
\label{section_comp}
The KTS algorithm is implemented in Matlab and is tested against a
number of problem instances with three equations and four variables of varying condition numbers.  Higher dimension problems, especially the ill-conditioned instances, require too much computation time due to the large number of hypercubes that must be considered.  We estimate the condition number by evaluating $\min \left\{ 1/\norm{f(x^0)}, \norm{\mpinv{f'(x^0)}} \right\}$ at the center point $x^0$ of every square considered by KTS during its execution and also at uniformly sampled points in $[0,1]^4$.  

Table \ref{table_res} compares the efficiency of KTS for each test
problem with its condition number.  The total number of hypercubes
examined by KTS during the entire computation, the width of the
smallest hypercube among those examined, and the maximum number of
Newton iterations to converge are reported.  Note that the high
number of Newton iterations of some test cases (the $9$th and $10$th rows of Table \ref{table_res}) is because the Jacobians of the zeros are ill-conditioned causing large roundoff error in the computation of the Newton iterations.
%

\begin{table}
\begin{center}
\begin{tabular}{|r|r|r|r|r|}
\hline \multirow{2}{*}{$\max \{m_1,m_2,m_3,m_4 \}$} &
\multicolumn{1}{|c|}{\multirow{2}{*}{$\cond{f}$}} & \multicolumn{1}{|c|}{Number of} & \multicolumn{1}{|c|}{Smallest} &
Max. Newton \\
& & hypercubes & \multicolumn{1}{|c|}{width} & \multicolumn{1}{|c|}{iterations} \\
\hline 
 2 & 6.60 & 641 & .03125 & 3 \\   
 2 & 11.5 & 3089 & .01563 & 3 \\ 
 2 & 15.5 & 673 & .03125 & 3 \\  
 3 & 24.0 & 145 & .06250 & - \\ 
 3 & 50.0 & 4273 & .00781 & 3 \\ 
 3 & 120 & 1009 & .00391 & 3 \\ 
 3 & $2.40 \times 10^3$ & 18177 & .00049 & 3 \\ 
 3 & $9.88 \times 10^4$ & 15841 & .00195 & 4 \\ 
 3 & $1.86 \times 10^7$ &  28881 & .00098 & 7 \\ 
 3& $2.66 \times 10^7$ & 29649 & .00098 & 7 \\ 
\hline
\end{tabular}
\end{center}
\caption{Efficiency of KTS algorithm on problems of different
condition numbers.\label{table_res}}
\end{table}

\section{Conclusion and future directions}
\label{section_conclusion}

We present the KTS algorithm for solving systems of polynomial equations with one more unknowns than the number of polynomials.  By using the combination of subdivision and Kantorovich's theorem, our algorithm can take advantage of the
quadratic convergence of Newton's method without the problems of
divergence and missing some solutions that commonly occur with
Newton's method. We also show that the efficiency of KTS has an
upper bound that depends solely on the condition number of the
problem instance.  Nevertheless, there are a number of questions left
unanswered by this article such as
\begin{itemize}
\item \textbf{Tighter bound on $r$.}
Some of the bounds in Section~\ref{section_analysis} appear
loose and could potentially underestimate the performance
of the algorithm.  For example, the scalars may be loose,
and one step in the argument preceding \eref{eq:rbdcase1}
uses the weak bound that $x^2\ge x$ since $x\ge 1$.  Thus, it
seems like there is room for tightening the analysis.  

A second
limitation of our analysis is that we establish a lower bound on
the smallest hypercube size, which indirectly places an upper bound on
the total number of hypercubes explored by the KTS algorithm (and hence its
running time).  This
upper bound, however, is usually far from tight as illustrated by
our computational experiments.  Thus, a different analysis that
addresses the number of hypercubes more directly would be useful.

\item \textbf{Using KTS in floating point arithmetic}. In the
presence of roundoff error, we may need to make adjustments for
KTS to be able to guarantee that the computed intersections are
accurate and that all of the solutions are found.

\item \textbf{Handling singular solutions and degenerate instances}.  Instances containing singular solutions or degeneracy are ill-posed, and our proposed KTS algorithm does not aim at handling such instances.  Certain applications, however, look for singular solutions or solutions to degenerate instances.  Further investigation on extending KTS to handle these situations would be beneficial.

\item \textbf{Other representations of $f$}.  As mentioned in the introduction, we
assume that $f$ is specified  by its $(m_1+1)(m_2+1)\cdots (m_{n+1}+1)$ 
Bernstein-B\'ezier control
points.  In many applications, however, there may be a more parsimonious representation.
For example, in the SSI problem, two surfaces of bi-degree $(m_1,m_2)$
are separately each represented by
$(m_1+1)(m_2+1)$ control points, hence $f=p_1-p_2$ is fully described
by $2(m_1+1)(m_2+1)$ control points rather than the
$(m_1+1)^2(m_2+1)^2$ control points needed for the general case.
It would be useful if the KTS algorithm could work 
directly on a more concise representation.

\item \textbf{Extension to general underdetermined polynomial systems}.  Polynomial systems with $n$ equations and $m > n+1$ unknowns generally contain higher dimension solutions.  The subdivision and exclusion test ideas still hold for the general case, but a different technique is needed to trace an approximation to the intersection surface.

\end{itemize}

\section{Acknowledgements}
We benefited from a helpful discussion with F.~Cucker about
condition numbers.

\bibliography{allbib}

\end{document}